\documentclass[10pt,conference]{IEEEtran}

\usepackage{amsmath,amssymb,epsfig,graphicx,subfigure}

\newtheorem{theorem}{Theorem}[section]

\newtheorem{lemma}[theorem]{Lemma}
\newtheorem{proposition}[theorem]{Proposition}

\newtheorem{definition}[theorem]{Definition}
\newtheorem{remark}[theorem]{Remark}

\newcommand\nc\newcommand
\nc\bfa{{\mathbf a}}
\nc\bfA{{\mathbf A}}\nc\cA{{\mathcal A}}
\nc\bfb{{\mathbf b}}\nc\bfB{{\mathbf B}}\nc\cB{{\mathcal B}}
\nc\bfc{{\mathbf c}}\nc\bfC{{\mathbf C}}\nc\cC{{\mathcal C}}
\nc\bfd{{\mathbf d}}\nc\bfD{{\mathbf D}}\nc\cD{{\mathcal D}}
\nc\bfe{{\mathbf e}}\nc\bfE{{\mathbf E}}\nc\cE{{\mathcal E}}
\nc\bff{{\mathbf f}}\nc\bfF{{\mathbf F}}\nc\cF{{\mathcal F}}
\nc\bfg{{\mathbf g}}\nc\bfG{{\mathbf G}}\nc\cG{{\mathcal G}}
\nc\bfh{{\mathbf h}}\nc\bfH{{\mathbf H}}\nc\cH{{\mathcal H}}
\nc\bfi{{\mathbf i}}\nc\bfI{{\mathbf I}}\nc\cI{{\mathcal I}}
\nc\bfj{{\mathbf j}}\nc\bfJ{{\mathbf J}}\nc\cJ{{\mathcal J}}
\nc\bfk{{\mathbf k}}\nc\bfK{{\mathbf K}}\nc\cK{{\mathcal K}}
\nc\bfl{{\mathbf l}}\nc\bfL{{\mathbf L}}\nc\cL{{\mathcal L}}
\nc\bfm{{\mathbf m}}\nc\bfM{{\mathbf M}}\nc\cM{{\mathcal M}}
\nc\bfn{{\mathbf n}}\nc\bfN{{\mathbf N}}\nc\cN{{\mathcal N}}
\nc\bfo{{\mathbf o}}\nc\bfO{{\mathbf O}}\nc\cO{{\mathcal O}}
\nc\bfp{{\mathbf p}}\nc\bfP{{\mathbf P}}\nc\cP{{\mathcal P}}
\nc\bfq{{\mathbf q}}\nc\bfQ{{\mathbf Q}}\nc\cQ{{\mathcal Q}}
\nc\bfr{{\mathbf r}}\nc\bfR{{\mathbf R}}\nc\cR{{\mathcal R}}
\nc\bfs{{\mathbf s}}\nc\bfS{{\mathbf S}}\nc\cS{{\mathcal S}}
\nc\bft{{\mathbf t}}\nc\bfT{{\mathbf T}}\nc\cT{{\mathcal T}}
\nc\bfu{{\mathbf u}}\nc\bfU{{\mathbf U}}\nc\cU{{\mathcal U}}
\nc\bfv{{\mathbf v}}\nc\bfV{{\mathbf V}}\nc\cV{{\mathcal V}}
\nc\bfw{{\mathbf w}}\nc\bfW{{\mathbf W}}\nc\cW{{\mathcal W}}
\nc\bfx{{\mathbf x}}\nc\bfX{{\mathbf X}}\nc\cX{{\mathcal X}}
\nc\bfy{{\mathbf y}}\nc\bfY{{\mathbf Y}}\nc\cY{{\mathcal Y}}
\nc\bfz{{\mathbf z}}\nc\bfZ{{\mathbf Z}}\nc\cZ{{\mathcal Z}}
\newcommand{\remove}[1]{}

\remove{
\nc\bsa{{\boldsymbol a}}\nc\bsA{{\boldsymbol A}}
\nc\bsb{{\boldsymbol b}}\nc\bsB{{\boldsymbol B}}
\nc\bsc{{\boldsymbol c}}\nc\bsC{{\boldsymbol C}}
\nc\bsd{{\boldsymbol d}}\nc\bsD{{\boldsymbol D}}
\nc\bse{{\boldsymbol e}}\nc\bsE{{\boldsymbol E}}
\nc\bsf{{\boldsymbol f}}\nc\bsF{{\boldsymbol F}}
\nc\bsg{{\boldsymbol g}}\nc\bsG{{\boldsymbol G}}
\nc\bsh{{\boldsymbol h}}\nc\bsH{{\boldsymbol H}}
\nc\bsi{{\boldsymbol i}}\nc\bsI{{\boldsymbol I}}
\nc\bsj{{\boldsymbol j}}\nc\bsJ{{\boldsymbol J}}
\nc\bsk{{\boldsymbol k}}\nc\bsK{{\boldsymbol K}}
\nc\bsl{{\boldsymbol l}}\nc\bsL{{\boldsymbol L}}
\nc\bsm{{\boldsymbol m}}\nc\bsM{{\boldsymbol M}}
\nc\bsn{{\boldsymbol n}}\nc\bsN{{\boldsymbol N}}
\nc\bso{{\boldsymbol o}}\nc\bsO{{\boldsymbol O}}
\nc\bsp{{\boldsymbol p}}\nc\bsP{{\boldsymbol P}}
\nc\bsq{{\boldsymbol q}}\nc\bsQ{{\boldsymbol Q}}
\nc\bsr{{\boldsymbol r}}\nc\bsR{{\boldsymbol R}}
\nc\bss{{\boldsymbol s}}\nc\bsS{{\boldsymbol S}}
\nc\bst{{\boldsymbol t}}\nc\bsT{{\boldsymbol T}}
\nc\bsu{{\boldsymbol u}}\nc\bsU{{\boldsymbol U}}
\nc\bsv{{\boldsymbol v}}\nc\bsV{{\boldsymbol V}}
\nc\bsw{{\boldsymbol w}}\nc\bsW{{\boldsymbol W}}
\nc\bsx{{\boldsymbol x}}\nc\bsX{{\boldsymbol X}}
\nc\bsy{{\boldsymbol y}}\nc\bsY{{\boldsymbol Y}}
\nc\bsz{{\boldsymbol z}}\nc\bsZ{{\boldsymbol Z}}
}

\nc\bsa{{\mathbf a}}\nc\bsA{{\mathbf A}}
\nc\bsb{{\mathbf b}}\nc\bsB{{\mathbf B}}
\nc\bsc{{\mathbf c}}\nc\bsC{{\mathbf C}}
\nc\bsd{{\mathbf d}}\nc\bsD{{\mathbf D}}
\nc\bse{{\mathbf e}}\nc\bsE{{\mathbf E}}
\nc\bsf{{\mathbf f}}\nc\bsF{{\mathbf F}}
\nc\bsg{{\mathbf g}}\nc\bsG{{\mathbf G}}
\nc\bsh{{\mathbf h}}\nc\bsH{{\mathbf H}}
\nc\bsi{{\mathbf i}}\nc\bsI{{\mathbf I}}
\nc\bsj{{\mathbf j}}\nc\bsJ{{\mathbf J}}
\nc\bsk{{\mathbf k}}\nc\bsK{{\mathbf K}}
\nc\bsl{{\mathbf l}}\nc\bsL{{\mathbf L}}
\nc\bsm{{\mathbf m}}\nc\bsM{{\mathbf M}}
\nc\bsn{{\mathbf n}}\nc\bsN{{\mathbf N}}
\nc\bso{{\mathbf o}}\nc\bsO{{\mathbf O}}
\nc\bsp{{\mathbf p}}\nc\bsP{{\mathbf P}}
\nc\bsq{{\mathbf q}}\nc\bsQ{{\mathbf Q}}
\nc\bsr{{\mathbf r}}\nc\bsR{{\mathbf R}}
\nc\bss{{\mathbf s}}\nc\bsS{{\mathbf S}}
\nc\bst{{\mathbf t}}\nc\bsT{{\mathbf T}}
\nc\bsu{{\mathbf u}}\nc\bsU{{\mathbf U}}
\nc\bsv{{\mathbf v}}\nc\bsV{{\mathbf V}}
\nc\bsw{{\mathbf w}}\nc\bsW{{\mathbf W}}
\nc\bsx{{\mathbf x}}\nc\bsX{{\mathbf X}}
\nc\bsy{{\mathbf y}}\nc\bsY{{\mathbf Y}}
\nc\bsz{{\mathbf z}}\nc\bsZ{{\mathbf Z}}

\nc\dist{{d_H}}
\nc\wt[1]{{|#1|}}
\def\s{\qopname\relax{no}{s}}
\nc{\veps}{\varepsilon} \nc{\ds}{\displaystyle}
\nc{\ts}{\textstyle}
\nc{\argmin}[1]{\underset{#1}{\operatorname{arg \hspace{0.25em} min}}}

\newcommand{\bP}[1]{{\mathbf{P}}\left[{#1}\right]}
\newcommand{\bbE}[2]{{\mathbf{E}}_{#1} {#2}}

\nc\qed{\mbox{\rule[0pt]{0.5ex}{1.3ex}}}
\renewcommand\phi\varphi

\begin{document}

\title{Two-Level Fingerprinting Codes}


\author{\authorblockN{N. Prasanth Anthapadmanabhan\authorrefmark{1}
and Alexander Barg\authorrefmark{1}\authorrefmark{2}}
\authorblockA{\authorrefmark{1}
Dept. of ECE and Inst. for Systems Research,
University of Maryland, College Park}
\authorblockA{\authorrefmark{2}
IPPI RAS, Moscow, Russia \\
Email:\{nagarajp, abarg\}@umd.edu
\vspace*{-1em}}
}
\remove{
\author{\authorblockN{N. Prasanth Anthapadmanabhan}
\authorblockA{Dept. of Electrical and Computer Eng. \\
University of Maryland, College Park \\
Email: nagarajp@umd.edu}
\and \authorblockN{Alexander Barg}
\authorblockA{Dept. of ECE and Inst. for Systems Research \\
University of Maryland, College Park \\
and IPPI RAS, Moscow, Russia\\ Email: abarg@umd.edu}
}
}
\maketitle

\begin{abstract}
We introduce the notion of {\em two-level} fingerprinting and traceability codes.
In this setting, the users are organized in a hierarchical manner 
by classifying them into various {\em groups}; for instance, 
by dividing the distribution area into several geographic regions, and 
collecting users from the same region into one group. 
Two-level fingerprinting and traceability codes have the following property: 
As in traditional (one-level) codes, when given an illegal copy produced by
a coalition of users, the decoder identifies one of the guilty 
users if the coalition size is less than a certain threshold $t$.
Moreover, even when the coalition is of a larger size $s$ $(> t)$, the decoder
still provides partial information 
by tracing one of the groups containing a guilty user. 

We establish sufficient conditions for a code to possess the two-level
traceability property. In addition, we also provide constructions for two-level
fingerprinting codes and characterize the corresponding set of achievable rates.
\end{abstract}

\section{Introduction}
In order to protect copyrighted digital content against unauthorized distribution
or {\em piracy}, several combinatorial schemes have been proposed in the
literature (see \cite{blackburn} for a survey). In this paper, we focus on two
such techniques: {\em fingerprinting codes} \cite{boneh} and 
{\em traceability codes} \cite{chor-fiat-naor-pin}.

The owner (distributor) of the content hides a unique mark called a
{\em fingerprint} in each licensed copy bought by a user. The collection
of fingerprint assignments is referred to as a {\em code}. 
If a naive user distributes a copy of his fingerprinted content illegally, 
then the pirated copy can easily be traced back to the guilty user. 
However, if a group of users ({\em pirates}) form a {\em coalition} to detect the
fingerprints and modify/erase them to create an illegal copy, 
then tracing a guilty user becomes a non-trivial task. 

Fingerprinting and traceability codes assign fingerprints in such a way that 
given an illegal copy, the distributor can use a tracing algorithm 
to identify at least one of the pirates as long as the coalition size does 
not exceed a certain threshold $t$, which is a parameter of the problem.
However, if the coalition size exceeds this threshold, the output of the tracing
algorithm can be useless.

To overcome this weakness, we formalize the notion of {\em multi-level} fingerprinting 
codes, which are inspired by error-correcting codes with unequal error protection 
used in communications problems (cf. for instance Bassalygo et al.
\cite{bassalygo}). We focus on the simplest case of two-level fingerprinting codes 
in this paper, but the concepts introduced apply to an arbitrary number of protection 
levels.

In this setting, the users are organized in a hierarchical manner, for instance, 
according to geographical location. The distribution area is divided into several 
regions, and users from the same region are collected into one {\em group}. 
The two-level fingerprinting 
codes studied in this paper have the following property: 
As in traditional (one-level) codes, the tracing algorithm determines at least
one of the guilty users if the coalition size is at most $t$.
Moreover, even when a larger number $s$ $(> t)$ of pirates participate, the 
algorithm provides partial information by retrieving the index of a group that
contains a member of the pirate coalition.

Formal definitions 
are available in Section \ref{sect:twolev-def}.
In Section \ref{sect:twolev-ta}, we obtain sufficient conditions for two-level 
traceability codes. Finally, we provide constructions for two-level fingerprinting 
codes and analyze the achievable rates in Section \ref{sect:twolev-fing}.

\section{Problem Statement} \label{sect:twolev-def}

Consider the problem where the content is to be distributed to $M_1M_2$ users
organized in $M_1$ groups, each of which contains $M_2$ users.
Assume that there is some ordering of the groups, and of the users within each group.
Thus, any user $\bsu$ is identified by a pair of indices 
$\bsu \equiv (u_1, u_2) \in [M_1] \times [M_2]$, where 
the notation $[n]$ stands for the set $\{1,\dots,n \}$.
For a user $\bsu=(u_1,u_2)$, let $\cG(\bsu)$ be its group index,
i.e., $\cG(\bsu)=u_1$.

The distributor hides a distinctive fingerprint in each legal copy.
The fingerprints are assumed to be distributed inside the host 
message so that their location is unknown to the users. The location of the 
fingerprints
is the same for all users.

Let $n$ denote the length of the fingerprints.
Let $\cQ$ denote an alphabet of (finite) size $q,$  
usually taken to be $\{0,\dots,q-1\}$ with modulo 
$q$ addition. An $(n,M)_q$ (one-level) {\em code} $(C,D)$ is a pair of
encoding and decoding mappings
$C: [M] \to \cQ^n$, $D: \cQ^{n} \to [M] \cup\{0\}$,
where the decoder output 0 signifies a decoding failure.
For convenience, we sometimes abuse terminology by calling
the range of $C$ a code, and use the same notation $C$ for it. 

The distributor's strategy of assigning fingerprints to 
users may be either deterministic or randomized as explained in the following 
subsections. Randomization can potentially
increase the number of users that can be supported for a given 
fingerprint length at the cost of a small error probability.

{\em Notation:} Throughout we will denote random variables (r.v.'s) by capital 
letters and their realizations by lower case letters. 
The Hamming distance between vectors $\bsx,\bsy$
will be written as $\dist(\bsx,\bsy)$, while $\wt \bsx$ denotes
the Hamming weight of $\bsx$. 
If $\cX$ is a set of vectors, we abbreviate 
$\min_{\bsx \in \cX} \dist(\bsx,\bsy)$ as $\dist(\cX,\bsy)$.
We will denote the $q$-ary entropy function by 
$h(x)= -x \log_q x/(q-1) -(1-x)\log_q(1-x)$.
For two functions $f(n),g(n),$ we write 
$f(n) \doteq g(n)$ if $\lim_{n \to \infty} n^{-1} \log (f(n)/g(n))=0$. 

\subsection{Deterministic Codes}

An $(n,M_1,M_2)_q$ {\em two-level code} $(C,D_1,D_2)$ is a triple 
consisting of one encoding and two decoding mappings
\begin{equation} 
\label{eqn:twolev-code}
\begin{array}{l}
C: [M_1] \times [M_2] \to \cQ^n, \vspace{0.3em}\\
D_1: \cQ^n \to [M_1] \cup \{0\}, \vspace{0.3em} \\
D_2: \cQ^n \to ([M_1] \times [M_2]) \cup \{0\},
\end{array}
\end{equation}
with 0 signifying a decoding failure. A two-level deterministic 
assignment of fingerprints is given by the encoding mapping $C$ of such a 
two-level code. 
The {\em rate pair} of an $(n,M_1,M_2)_q$ two-level code is defined as 
$$(R_1,R_2):=\left(\frac{1}{n}\log_q M_1,\frac{1}{n}\log_q M_2 \right).$$

A \emph{coalition} of users is an arbitrary subset of $[M_1] \times [M_2]$. 
Members of the coalition are commonly referred to as {\em pirates}. 
A coalition $U$ has access to the collection of fingerprints,
namely $C(U)$, that are assigned to it. 
Let $U$ be a coalition of $t$ users and suppose
$C(U) = \{\bsx_1,\dots,\bsx_t\}$.
In order to conceal their identities from the distributor, the
coalition's members attempt to create a pirated copy with a
modified fingerprint $\bsy \in\cQ^{n}$. We assume that the code 
$(C,D_1,D_2)$ is public and can be used by the pirates in designing their 
attack. 

Note that although the fingerprint locations are not available to the pirates, 
they may detect some of these locations by comparing their copies
for differences and modify the detected positions. 
Coordinate $i$ of the fingerprints is called {\em undetectable} for the
coalition $U$ if $x_{1i}=x_{2i}=\dots=x_{ti}$ and is called 
{\em detectable} otherwise.
The set of forgeries that can be created by the coalition in this manner
is called the \emph{envelope} and is given by:
\begin{equation}  \label{eqn:twolev-env}
\cE(\bsx_1,\dots,\bsx_t)= \left\{\bsy \in \cQ^n~\big|~y_i 
\in \{x_{1i},\dots,x_{ti}\}, 
\forall i \in [n] \right\}.
\end{equation}

Given a pirated copy with a forged fingerprint, the distributor performs 
tracing based on $D_1$ and $D_2$ to locate one of the pirates.
The decoder $D_2$ attempts to trace the exact identity of one of the pirates, 
while $D_1$ focuses only on locating a group containing at least one of the pirates. 

In order to extend the notion of traceability to two-level codes, let us 
consider the case where the tracing is accomplished using
{\em minimum distance} (MD) decoding. Specifically, we take
\begin{equation}
\label{eqn:twolev-MDdec}
\begin{array}{l}
D_2(\bsy) = \argmin{\bsu \in [M_1] \times [M_2]} \dist(C(\bsu),\bsy),\\
D_1(\bsy) = \cG(D_2(\bsy)).
\end{array}
\end{equation}
If the minimum distance above is attained for multiple users, the decoder
$D_2$ outputs any one of the closest users.
This leads us to the notion of two-level traceability codes in the 
deterministic setting.

\begin{definition} 
\label{def:twolev-ta}
A two-level code $C$ has {\em $(t_1,t_2)$-traceability} property 
(or is $(t_1,t_2)$-TA) where $t_1 > t_2$ if:
\begin{itemize}
\item[(a)] For any coalition $U$ of size at most $t_2$ and any
$\bsy \in \cE(C(U))$, the decoding result $D_2(\bsy) \in U$.
\item[(b)] For any coalition $U$ of size at most $t_1$ and any
$\bsy \in \cE(C(U))$, the decoding result $D_1(\bsy) \in \cG(U)$.
\end{itemize}
\end{definition}

We observe that an $(n,M_1,M_2)_q$ two-level code which is $(t_1,t_2)$-TA 
has the $t_2$-TA property when viewed as an $(n,M_1M_2)_q$ one-level code;
moreover, for coalitions of the larger size $t_1$, one of the groups
containing a pirate is closer to the forgery compared to the remaining
groups.
In this paper, we examine sufficient conditions under which a two-level 
code has the $(t_1,t_2)$-traceability property.

\subsection{Randomized Codes}

A randomized strategy to assign fingerprints is defined as 
the following random experiment.
The distributor has a family of $(n,M_1,M_2)_q$ two-level codes 
$\{(C_k,D_{1k},D_{2k}), k \in \cK\}$, where $\cK$ is a finite set
of ``keys''. 
The distributor chooses one of the keys according to a probability distribution 
$(\pi(k), k \in \cK)$. If the key $k$ is selected, then fingerprints are 
assigned according to $C_k$ and tracing is done using $D_{1k}$ and $D_{2k}$. 
The code resulting from this random experiment is called a (two-level) 
{\em randomized code} and is denoted by $(\cC,\cD_1,\cD_2)$.   

Following the standard convention in cryptography of the system design
being publicly available, we allow the users to have knowledge of the family of codes
$\{(C_k,D_{1k},D_{2k})\}$ and the distribution $\pi(\cdot)$, while the exact key 
choice is kept secret by the distributor. 

Consider a coalition $U$ of size $t$. Any attack by the coalition 
can be modeled as a randomized strategy $V(\cdot|\cdot,\dots,\cdot)$, 
where $V(\bsy|\bsx_{1},\dots,\bsx_{t})$ 
gives the probability that the coalition creates $\bsy$ given that it observes 
the fingerprints $\bsx_{1},\dots,\bsx_{t}$. 
Our interest is in a special class of strategies which satisfy
the restrictions (\ref{eqn:twolev-env}) in creating a forgery. 
A strategy $V$ is called {\em admissible} if
$$
V(\bsy|\bsx_{1},\dots,\bsx_{t}) = 0 \text{ for all } 
\bsy \notin \cE(\bsx_{1},\dots,\bsx_{t}).
$$
Let $\cV_{t}$ denote the class of admissible strategies. 

Denote the random forgery generated by $U$ using the strategy $V$ by 
$\bsY_{\cC,U,V}$. The distributor, on observing the forged fingerprint, 
employs the decoders $D_{1k}$ and $D_{2k}$ while using the key $k$. 
For a given coalition $U$ and strategy $V$, we define the following 
error probabilities: 
\begin{align*} 
e_1(\cC,\cD_1,U,V) & = \bP{ \cD_1(\bsY_{\cC,U,V}) \notin \cG(U) } \\
	& = \bbE{}{\sum_{\bsy:D_{1K}(\bsy) \notin \cG(U)} V(\bsy|C_K(U))}, \\
e_2(\cC,\cD_2,U,V) & = \bP{ \cD_2(\bsY_{\cC,U,V}) \notin U } \\
 & = \bbE{}{\sum_{\bsy:D_{2,K}(\bsy) \notin U} V(\bsy|C_K(U))},
\end{align*}
where the expectation is over the r.v. $K$ with distribution $\pi(k)$. 

\begin{definition} 
\label{def:twolev-fing} 
A randomized code $(\cC,\cD_1,\cD_2)$ is said to be a 
$(t_1,t_2)$-\emph{fingerprinting with $\veps$-error} 
where $t_1 > t_2$ if: 
\begin{itemize}
\item[(a)] For any coalition $U$ of size at most $t_2$ and any 
admissible strategy $V$, the error probability $e_2(\cC,\cD_2,U,V) \le \veps$.
\item[(b)] For any coalition $U$ of size at most $t_1$ and any 
admissible strategy $V$, the error probability $e_1(\cC,\cD_1,U,V) \le \veps$.
\end{itemize}
\end{definition}

We observe that an $(n,M_1,M_2)_q$ two-level code which is 
$(t_1,t_2)$-fingerprinting has the $t_2$-fingerprinting property when viewed as 
an $(n,M_1M_2)_q$ one-level code; in addition, 
for the larger size-$t_1$ coalitions,
the tracing algorithm can locate a group containing one of the pirates 
with high probability.

A rate pair $(R_1,R_2)$ is said to be {\em achievable} for $q$-ary 
$(t_1,t_2)$-fingerprinting if there exists a sequence of 
$(n,q^{nR_{1n}},q^{nR_{2n}})_q$ randomized codes that are 
$(t_1,t_2)$-fingerprinting with error probability $\veps_n$
such that
$$ \lim_{n \to \infty} \veps_n =0, \qquad 
\liminf_{n \to \infty} R_{in} = R_i, \quad i=1,2.$$ 

The goal of this paper is to investigate constructions of two-level 
fingerprinting codes and to characterize the corresponding set of achievable 
rate pairs.

\begin{remark} \label{rem:twolev-obs}
\begin{enumerate}
\item If an $(n,M_1,M_2)_q$ two-level code is $(t_1,t_2)$-fingerprinting 
(resp., TA), then choosing any single user from every 
group forms an $(n,M_1)_q$ one-level code that is $t_1$-fingerprinting (resp., TA).
\item If an $(n,M_1M_2)_q$ one-level code is $t_1$-fingerprinting
(resp., TA), then for any $t_2 < t_1$, it can also be treated as a 
$(n,M_1,M_2)_q$ two-level code that is $(t_1,t_2)$-fingerprinting 
(resp., TA).
\end{enumerate}
\end{remark}

\section{Traceability Codes} \label{sect:twolev-ta}

It is known \cite{chor-fiat-naor-pin} that a one-level 
code of length $n$ is $t$-TA if the distance between
any pair of fingerprints is strictly greater than $n(1-1/t^2)$.
We wish to obtain an analogous result for the case of two-level codes. 

For a given two-level code $C$, we define the following minimum distances:
\begin{align}
d_1(C) & := \min_{\substack{
				\bsu,\bsv \in [M_1] \times [M_2]\\
				u_1 \neq v_1}} \dist(C(\bsu),C(\bsv)), 
				\label{eqn:twolev-d1} \\
d_2(C) & := \min_{\substack{
				\bsu,\bsv \in [M_1] \times [M_2]\\
				u_2 \neq v_2}} \dist(C(\bsu),C(\bsv)).
				\label{eqn:twolev-d2}
\end{align}
Let $d(C) = \min \left( d_1(C),d_2(C) \right)$. 

\begin{proposition} \label{prop:twolev-ta}
Suppose $t_1 > t_2$ and $C$ is a two-level code of length $n$ with
$d_1(C)> n(1-1/t_1^2) \quad \text{and} \quad d_2(C) > n(1-1/t_2^2).$
Then $C$ is $(t_1,t_2)$-TA.
\end{proposition}
\begin{proof}
It is straightforward to see that the assumptions in the proposition imply that 
$d(C) > n(1-1/t_2^2)$.  Therefore, property (a) in Definition \ref{def:twolev-ta} 
follows directly from the result for one-level codes.

Next, we show that property (b) is a consequence of $d_1(C)> n(1-1/t_1^2)$.
Let $U$ be a coalition of size at most $t_1$ and $\bsy \in \cE(C(U))$.
Then, there exists some user $\bsu \in U$ who coincides with $\bsy$ in at least 
$n/t_1$ coordinates. For any user $\bsu^\prime$ such that $\cG(\bsu^\prime) \notin \cG(U)$,
the number of agreements with $\bsy$ is at most $t_1(n-d_1(C)) < \frac n{t_1}$, 
thus establishing property (b).
\end{proof}

\remove{As a consequence, ideas used for constructing unequal error protection codes 
can be used to construct two-level traceability codes. As in the
one-level case, to be able to construct large-size codes based on the
above sufficient conditions, one needs large alphabets. }

\section{Fingerprinting Codes} \label{sect:twolev-fing}

For $w \in [n]$, denote $\cS_{w,n} := \{\bsx \in \cQ^n: \wt \bsx = w \}$. 
For $R_1,R_2 \in [0,1]$,
define $M_{1n}=\lfloor q^{nR_1} \rfloor$, $M_{2n}= \lfloor q^{nR_2} \rfloor$.
Fix $\omega \in [0,1]$. We take $n$ such that $w=\omega n$ is an integer
and construct an $(n,M_{1n},M_{2n})_q$ 
two-level randomized code $(\cC^{\omega}_{n},\cD^{\omega}_{1n},\cD^{\omega}_{2n})$
as follows. 

For $i \in [M_{1n}]$, pick vectors $\bsR_i$ independently and uniformly 
at random from $\cQ^n$. We will refer to the $\bsR_i$'s as ``centers''.
Choose $\bsS_{ij},~(i,j) \in [M_{1n}] \times [M_{2n}]$, independently and uniformly 
at random from $\cS_{w,n}$. Generate $M_{1n}M_{2n}$ fingerprints
$$\bsX_{ij} = \bsR_i + \bsS_{ij}, \quad (i,j) \in [M_{1n}] \times [M_{2n}]$$
and assign $\bsX_{ij}$ as the fingerprint for user $(i,j)$. 

Once the fingerprints are assigned, tracing is based on the 
MD decoder (\ref{eqn:twolev-MDdec}). The MD decoder may be sub-optimal in 
general; however, it is amenable for analysis in our construction. 

In the following subsections, we analyze the error probability and
characterize the achievable rate pairs 
for the above construction. 
The lemmas below will be useful in the analysis.

\begin{lemma} \label{lem:twolev-indep}
Let $\bsS$ have a uniform distribution on $\cS_{w,n}$. Then, for $l \in [n]$
and $a \in \cQ \backslash \{0\}$, $\bP{S_l=a} = \omega/(q-1)$. Moreover, 
the r.v.'s $\{S_l, l \in [n] \}$ are asymptotically pairwise independent.
\end{lemma}

\begin{lemma} \label{lem:twolev-chebyshev}
Fix $p \in [0,1]$ and $\veps >0$. For $l \in [n]$, let $Z_l$ be a Bernoulli r.v.
with $\bP{Z_l = 1} = p$, and let $\{Z_l, l \in [n]\}$ be pairwise independent.
Then, with $Z:= \sum_{l \in [n]} Z_l$, we have 
$$\bP{Z \notin [n(p-\veps),n(p+\veps)]} \le \frac{p(1-p)}{\veps^2 n}.$$
\end{lemma}

{\em Notation:} For a coalition $U =\{ \bsu^1, \dots, \bsu^t \}$, 
we denote the realizations of $\bsX_{\bsu^i}, \bsR_{u_1^i},\bsS_{\bsu^i}$
by $\bsx_i,\bsr_i,\bss_i$ respectively, with $\bsx_i = \bsr_i + \bss_i, i \in [t]$. 
Let $\bsz \in \cQ^t$ be a vector. Denote by $s_{\bsz}(\bsx_1,\dots,\bsx_t)$ the number
of columns equal to $\bsz^T$ in the matrix whose rows are $\bsx_1,\dots,\bsx_t$. 
For $p \in [0,1]$ and $\veps >0$, define
$I_n(p,\veps) := [n(p-\veps),n(p+\veps)].$

\subsection{$(t,1)$-fingerprinting}

First, we consider the $(2,1)$-fingerprinting property. 
This is the simplest case of two-level fingerprinting that goes beyond 
the known techniques for one-level codes. Although coalitions of size 1 are 
trivial to handle for one-level fingerprinting, it is still non-trivial to 
construct a $(2,1)$-fingerprinting code. 

\begin{theorem} \label{thm:twolev-21}
\remove{
$(R_1,R_2)$ is an achievable rate pair for $q$-ary $(2,1)$-fingerprinting 
if there exists an $\omega \in [0,(q-1)/2q]$ such that
}
For any $\omega \in [0,(q-1)/2q]$, the randomized code 
$(\cC^{\omega}_{n},\cD^{\omega}_{1n},\cD^{\omega}_{2n})$
is $(2,1)$-fingerprinting with error probability decaying to 0 if
\begin{align}
R_1 & < 1-h ((q-1)/2q + \omega), 
\label{eqn:twolev-21-1}\\
R_2 & < h(\omega). \label{eqn:twolev-21-2}
\end{align}
\end{theorem}

{\em Discussion:} The above theorem provides a set of achievable
rate pairs for $q$-ary $(2,1)$-fingerprinting. Let us fix $\cQ = \{0,1\}$ and 
put the result in the perspective of bounds available for one-level fingerprinting 
(see Figure \ref{fig:twolev-graph}).

\begin{figure}[ptb]
\centering \includegraphics[width=3.6in]{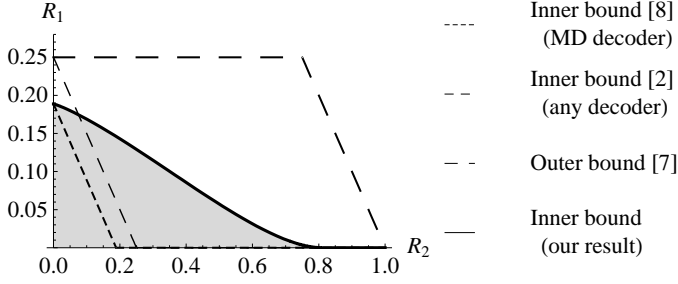}
\caption[Achievable rate region for binary $(2,1)$-fingerprinting]
{Achievable rate region for binary $(2,1)$-fingerprinting. The bounds 
from previous works follow by using one-level fingerprinting schemes.}
\label{fig:twolev-graph}
\end{figure}

\begin{itemize}
\item {\em Outer bound}: Since the $(2,1)$-fingerprinting property implies
one-level 1-fingerprinting, we should have $R_1 + R_2 \le 1$. Moreover, 
$R_1$ cannot exceed the rate of a one-level 2-fingerprinting code
(by part (1) of Remark \ref{rem:twolev-obs}); thus, any upper bound for it also 
applies to $R_1$. 
In particular, by \cite{huang-moulin} $R_1 \le 0.25$.

\item {\em Inner bound}: By part (2) of Remark \ref{rem:twolev-obs}, 
the rate pairs $(R_1,R_2)$ such that $R_1 + R_2 < 0.188$ 
are achievable (with MD decoding) using the $2$-fingerprinting code
given in \cite{elgamal}. In fact, by allowing other decoders, we can do better,
achieving $R_1 + R_2 < 0.25$ through the $2$-fingerprinting construction in \cite{abd}.
\end{itemize}

\begin{proof}(of Theorem \ref{thm:twolev-21})
{\em Size-1 coalitions:} Let $\bsu = (u_1,u_2)$ be the pirate. 
For size-1 coalitions, the envelope is degenerate as it consists of only the user's own 
fingerprint. Now,
\begin{align*}
& e_2(\cC^{\omega}_{n},\cD^{\omega}_{2n},\bsu) \\
& =  \bP{\exists \bsu^{\prime} \neq \bsu: \bsX_{\bsu^{\prime}} = \bsX_{\bsu}} \\
& \le \bP{\exists \bsu^{\prime} \neq \bsu:  u_1^{\prime}=u_1, \bsX_{\bsu^{\prime}} = \bsX_{\bsu}} \\
& \qquad + \bP{\exists \bsu^{\prime} \neq \bsu: u_1^{\prime} \neq u_1, \bsX_{\bsu^{\prime}} = \bsX_{\bsu}} \\
& \overset{(a)}{\le} \bP{\exists u_2^{\prime} \neq u_2: \bsS_{u_1u_2^{\prime}} = \bsS_{u_1u_2}} \\
& \qquad + \bP{\exists u_1^{\prime} \neq u_1: \dist(\bsR_{u_1^{\prime}},\bsX_{\bsu}) \le w} \\
& \overset{(b)}{\le} q^{nR_2} \bP{\bsS_{u_1u_2^{\prime}} = \bsS_{u_1u_2}} 
+ q^{nR_1} \bP{\dist(\bsR_{u_1^{\prime}},\bsX_{\bsu}) \le w} \\
& \doteq q^{-n(h(\omega) - R_2)} + q^{-n(1-h(\omega) - R_1)},
\end{align*}
where (a) is due to the fact that if the fingerprint of another user 
matches with the pirate's fingerprint, then the corresponding center is within distance 
$w$ from the pirate's fingerprint, and (b) follows from the union bound.
Consequently, the error probability for size-1 coalitions approaches 0 if 
$R_2 < h(\omega)$ and $R_1 < 1-h(\omega)$.

{\em Size-2 coalitions:} There are two possibilities: either both users are in the
same group or they are in different groups. It turns out that the latter case
is the dominant one. Since the analysis for the two cases is similar, we only consider the 
latter case below. 

Let $U=\{ \bsu^1,\bsu^2 \}$ be such a coalition.
For any strategy $V \in \cV_2$, we have
\begin{align}
& e_1(\cC^{\omega}_{n},\cD^{\omega}_{1n},U,V) \notag \\
& = \sum_{\bsr_1,\bsr_2,\bss_1,\bss_2} \bP{
\bsr_1,\bsr_2,\bss_1,\bss_2} \sum_{\bsy} V(\bsy|\bsx_1,\bsx_2) \notag \\
& \quad \qquad \times \bP{\cD^{\omega}_{1n}(\bsy) \notin \cG(U) \Big| 
\bsr_1,\bsr_2,\bss_1,\bss_2
}. 
\label{eqn:twolev-21fing}
\end{align}
Consider the inner probability term
\begin{align*}
& \bP{\cD^{\omega}_{1n}(\bsy) \notin \cG(U) \Big| 
\bsr_1,\bsr_2,\bss_1,\bss_2
} \\
& \overset{(a)}{\le} \bP{ \exists \bsu^{\prime} \notin U: u_1^{\prime} \notin \cG(U),
\dist(\bsX_{\bsu^{\prime}},\bsy) \le \dist(\{\bsx_1,\bsx_2\},\bsy) } \\
& \overset{(b)}{\le} \bP{ \exists u_1^{\prime} \notin \cG(U):
\dist(\bsR_{u_1^{\prime}},\bsy) \le \dist(\{\bsx_1,\bsx_2\},\bsy) + w} 
\end{align*}\begin{align*}
& \le q^{nR_1} \bP{\dist(\bsR_{u_1^{\prime}},\bsy) \le 
\dist(\{\bsx_1,\bsx_2\},\bsy) + w},
\end{align*}
where we have exploited the independence in the construction in (a), and
(b) follows because if the fingerprint of another user is within distance $d$
from $\bsy$, then the corresponding center is within $d+w$ from $\bsy$. 
For $\veps > 0$, define 
$$\cT^{\veps}_n:= \left \{ (\bsr_1,\bsr_2,\bss_1,\bss_2): \begin{array}{l}
\forall a \in \cQ, \\
\s_{(a,a)}(\bsx_1,\bsx_2) \in I_n(1/q^2,\veps/q)
\end{array} \right \}.$$
Observe that $\bsX_{\bsu^1}$ and $\bsX_{\bsu^2}$ are independent and uniformly 
distributed over $\cQ^n$.
Therefore, using Lemma \ref{lem:twolev-chebyshev}, it is a simple matter to show that  
$\bP{(\bsR_{u_1^1},\bsR_{u_2^1},\bsS_{\bsu^1},\bsS_{\bsu^2}) \notin \cT^{\veps}_n }$ decays to 0
as $n \to \infty$. Now, take any $(\bsr_1,\bsr_2,\bss_1,\bss_2) \in \cT^{\veps}_n$
and $\bsy \in \cE(\bsx_1,\bsx_2)$. The number of undetectable positions in 
$\{\bsx_1,\bsx_2\}$ is at least $n(1/q - \veps)$, implying that 
$\dist(\{\bsx_1,\bsx_2\},\bsy) \le \frac{n}{2}\left(1-\frac{1}{q} +\veps \right)$. 
Thus, in this case
\begin{align*}
& q^{nR_1} \bP{\dist(\bsR_{u_1^{\prime}},\bsy) \le \dist(\{\bsx_1,\bsx_2\},\bsy) + w} \\
& \le q^{nR_1} \bP{\dist(\bsR_{u_1^{\prime}},\bsy) \le \frac{n}{2}\left(1-\frac{1}{q} +\veps \right) + w} \\
& \doteq q^{-n(1-h \left( \frac{1}{2}\left(1-\frac{1}{q} +\veps \right) + \omega \right) -R_1)}.
\end{align*}
Substituting the above in (\ref{eqn:twolev-21fing}) and taking $\veps \to 0$, 
we conclude that the error probability for size-2 coalitions approaches 0 if 
(\ref{eqn:twolev-21-1}) holds. 
\end{proof}

We now extend the techniques to larger coalitions.

\begin{theorem} \label{thm:twolev-t1}
\remove{
$(R_1,R_2)$ is an achievable rate pair for $q$-ary $(t,1)$-fingerprinting if
there exists an $\omega$ with
$\frac{t-1}{t} \big(1 - \frac{1}{q^{t-1}} \big) + \omega \le \frac{q-1}{q}$,
such that
}
For any $\omega$ such that 
$\frac{t-1}{t} \big(1 - \frac{1}{q^{t-1}} \big) + \omega \le \frac{q-1}{q}$, 
the randomized code $(\cC^{\omega}_{n},\cD^{\omega}_{1n},\cD^{\omega}_{2n})$
is $(t,1)$-fingerprinting with error probability decaying to 0 if
\begin{align}
R_1 & < 1-h \left(\frac{t-1}{t} \Big(1 - \frac{1}{q^{t-1}} \Big) + \omega \right), 
\label{eqn:twolev-t1-1}\\
R_2 & < h(\omega). \label{eqn:twolev-t1-2}
\end{align}
\end{theorem}

\begin{proof}
{\em Size-$1$ coalitions:} For a single pirate $\bsu$, the analysis in 
Theorem \ref{thm:twolev-21} proves that the probability of decoding error
approaches 0 if $R_2 < h(\omega)$ and $R_1 < 1-h(\omega)$.

{\em Size-$t$ coalitions:} It can be shown that the case where the $t$ pirates 
are in distinct groups is the dominant one. Once this is shown, we use exactly 
the same arguments as in the case of size-2 coalitions 
in Theorem \ref{thm:twolev-21}. We finally obtain that the error probability
for coalitions of size $t$ approaches 0 if (\ref{eqn:twolev-t1-1}) holds.
\end{proof}

\begin{remark} \label{rem:twolev-t1}
A sufficiently large alphabet is required in order for an $\omega$ satisfying
$\frac{t-1}{t} \left(1 - \frac{1}{q^{t-1}} \right) + \omega \le \frac{q-1}{q}$ to 
exist. For instance, it suffices to take $q\ge t+1.$
\end{remark}

\subsection{$(t,2)$-fingerprinting}
Let $q \ge 3$. For $\omega,\gamma,\alpha,\beta \in [0,1]$, with 
$\alpha \le 1-\gamma$, $\beta \le \gamma$, $\alpha + \beta \le \omega$, 
$\omega - \alpha \le \gamma$, let 
\begin{align*}
&\phi(\omega,\gamma,\alpha,\beta) \\
& := (1-\gamma) h\left(\frac{\alpha}{1-\gamma}\right)
+ (\gamma - \beta) h \left( \frac{\omega - \alpha - \beta}{\gamma -\beta} \right) \\
& \quad + \gamma h\left(\frac{\beta}{\gamma} \right) 
+ (\omega -\alpha) \log_q \left( \frac{q-2}{q-1} \right) 
- \beta \log_q(q-2).
\end{align*}
Let
\begin{align*}
\delta_1(\omega) & = \frac{1}{2} \left( 1-(1-\omega)^2-\frac{\omega^2}{q-1} \right), \\
\delta_2(\omega) & = \frac{1}{2} \left( 1-\frac{1}{q} \right), \\ 
{f}_1(\omega) &= \max_{\substack{\gamma,\alpha,\beta: \\
\omega^2 \le \gamma \le 1-(1-\omega)^2,\gamma - \beta + \alpha \le \delta_1(\omega)
}}\hspace*{-2mm}
 \phi(\omega,\gamma,\alpha,\beta), \\
{f}_2(\omega) &= 
\max_{\substack{
\gamma,\alpha,\beta: \\
\omega \left( \frac{q-1}{q} \right) \le \gamma \le 1-\frac{1-\omega}{q},
\gamma - \beta + \alpha \le \delta_2(\omega)
}}
\hspace*{-0.5em}
 \phi(\omega,\gamma,\alpha,\beta).
\end{align*}
\begin{theorem}
\remove{
Let $q \ge 3$. $(R_1,R_2)$ is an achievable rate pair for $q$-ary $(t,2)$-fingerprinting 
if there exists an $\omega$ with
$\frac{t-1}{t} \Big(1 - \frac{1}{q^{t-1}} \Big) + \omega \le \frac{q-1}{q}$,
such that
}
Let $q \ge 3$. For any $\omega$ such that 
$\frac{t-1}{t} \big(1 - \frac{1}{q^{t-1}} \big) + \omega \le \frac{q-1}{q}$, 
the randomized code $(\cC^{\omega}_{n},\cD^{\omega}_{1n},\cD^{\omega}_{2n})$
is $(t,2)$-fingerprinting with error probability decaying to 0 if
\begin{align}
R_1 & < 1-h \left(\frac{t-1}{t} \Big(1 - \frac{1}{q^{t-1}} \Big) + \omega \right), 
\label{eqn:twolev-t2-1}\\
R_2 & < h(\omega) - \max( {f}_1(\omega), {f}_2(\omega) ). \label{eqn:twolev-t2-2}
\end{align}
\end{theorem}
\begin{proof}
Size-$t$ coalitions are handled in the same way as in Theorem \ref{thm:twolev-t1}.

{\em Size-$2$ coalitions:} 
There are two possibilities depending on whether the pirates
belong to the same group or not. We sketch the case where they are in different groups below.
The other case is analyzed similarly.

Consider a coalition $U = \{ \bsu^1, \bsu^2 \}$, where the users are in different groups,
and let $V \in \cV_2$ be an admissible strategy. We have
\begin{align}
& e_2(\cC^{\omega}_{n},\cD^{\omega}_{2n},U,V) \notag\\
& = \sum_{\bsr_1,\bsr_2,\bss_1,\bss_2} \bP{\bsr_1,\bsr_2,\bss_1,\bss_2} 
\sum_{\bsy} V(\bsy|\bsx_1,\bsx_2) \notag \\
& \quad \qquad \times 
\bP{\cD^{\omega}_{2n}(\bsy) \notin U | \bsr_1,\bsr_2,\bss_1,\bss_2}.
\end{align}
Now,
$$ \left [\cD^{\omega}_{2n}(\bsy) \notin U \right] = E_1 \cup E_2 \cup E_3, $$
where, the events $E_1,E_2,E_3$ are formed of those $\bsu^{\prime} \notin U$
that satisfy $\dist(\bsX_{\bsu^{\prime}},\bsy) \le 
\dist(\{\bsx_1,\bsx_2\},\bsy)$ and the conditions $u_1^{\prime} = u_1^1$, 
$u_1^{\prime} = u_1^2$, $u_1^{\prime} \notin \cG(U),$ respectively.
The error event $E_3$ was already analyzed in Theorem \ref{thm:twolev-21} 
and its conditional probability approaches 0 if (\ref{eqn:twolev-21-1}) holds. 
We consider $E_1$ below. The analysis for $E_2$ is identical by symmetry.
\begin{align}
& \bP{ E_1 \Big | 
\bsr_1,\bsr_2,\bss_1,\bss_2} \notag\\
& = \bP{ \exists u_2^{\prime} \neq u_2^1: 
\dist(\bsr_1+\bsS_{u_1^1 u_2^{\prime}},\bsy) \le \dist(\{\bsx_1,\bsx_2\},\bsy)} \notag \\
& \le q^{nR_2} \bP{\dist(\bsr_1+\bsS_{u_1^1 u_2^{\prime}},\bsy) \le \dist(\{\bsx_1,\bsx_2\},\bsy)} \notag\\
& = q^{nR_2} \bP{\dist(\bsS_{u_1^1 u_2^{\prime}},\bsy^{\prime}) \le 
\dist(\{\bss_1,\bsr_1+\bsx_2\},\bsy^{\prime}) }, \label{eqn:twolev-t2-3}
\end{align}
where $\bsy^{\prime}=\bsy + \bsr_1 \in \cE(\bss_1,\bsr_1+\bsx_2)$.
In this case, we use Lemmas \ref{lem:twolev-indep} and \ref{lem:twolev-chebyshev} to
show that 
$$\cT^{\veps}_n:= \left \{ (\bsr_1,\bsr_2,\bss_1,\bss_2): \begin{array}{l}
\s_{(0,0)}(\bss_1,\bsr_1+\bsx_2) \simeq n\frac{1-\omega}{q} \\
\s_{(a,a^{\prime})}(\bss_1,\bsr_1+\bsx_2) \simeq n \frac{\omega}{(q-1)q} \\
\forall a,a^{\prime} \in \cQ \backslash \{0\}
\end{array} \right \}.$$
is the typical set. For simplicity, we have omitted $\veps$ and will use the 
approximate relations 
$\simeq$, $\lesssim$, $\gtrsim$ in its place.
Now, take any $(\bsr_1,\bsr_2,\bss_1,\bss_2) \in \cT^{\veps}_n$ 
and $\bsy^{\prime} \in \cE(\bss_1,\bsr_1+\bsx_2)$. The number of undetectable 
positions in $\{\bss_1,\bsr_1+\bsx_2\}$ is $\simeq n/q$, while the number of
coordinates where both symbols are non-zero is $\simeq n \omega (q-1)/q$. This implies 
$\dist(\{\bss_1,\bsr_1+\bsx_2\},\bsy^{\prime}) \lesssim n \delta_2(\omega)$
and $ n \omega (q-1)/q \lesssim \wt{\bsy^{\prime}} \lesssim n(1- (1-\omega)/q)$.

Let $\wt{\bsy^{\prime}}=\gamma n$, where $\gamma \in [0,1]$.
Then
$$
\bP{\dist(\bsS_{u_1^1 u_2^{\prime}},\bsy^{\prime}) \le n \delta_2(\omega) } 
\doteq q^{-n E(\omega,\gamma)},
$$
where 
$$E(\omega,\gamma) = h(\omega) - \max_{\substack{
\alpha,\beta: \\
\gamma - \beta + \alpha \le \delta_2(\omega)
}} \phi(\omega,\gamma,\alpha,\beta). $$
Since $\gamma$ can be chosen by the pirates such that 
$\omega \frac{q-1}q \lesssim \gamma \lesssim 1- \frac{1-\omega}{q}$,
by substituting the above in (\ref{eqn:twolev-t2-3}), we conclude that the
conditional probability of $E_1$ (and $E_2$) approaches 0 if 
$R_2 < h(\omega) - f_2(\omega)$. Similarly, we obtain
$R_2 < h(\omega) - f_1(\omega)$ when the pirates are in the same group. 
\end{proof}

Let us show that the rate region thus defined is nontrivial. 
Given $\omega$ and $\gamma,$ the maximizing values of the other arguments of $\phi$ are $\alpha= \omega(1-\gamma)$ and
$\beta= \omega \gamma/(q-1)$, so
\begin{align*}
\phi(\omega,\gamma,\alpha,\beta) \le h(\omega) - \gamma \omega 
\Big(\log_q  \frac{q-1}{q-2}  +  \frac{\log_q(q-2)}{q-1} \Big).
\end{align*}
Consequently, we get 
$\max(f_1(\omega),f_2(\omega)) \le h(\omega) - 
D$, where 
$D=D(\omega)=\omega^3 \big(\log_q  \frac{q-1}{q-2}  + \frac{\log_q(q-2)}{q-1} \big)$ 
and $D(\omega)>0$ for all $\omega>0.$
This shows that the r.-h.s. of (\ref{eqn:twolev-t2-2}) is positive. 
By Remark \ref{rem:twolev-t1}, the r.-h.s. of (\ref{eqn:twolev-t2-1}) is also positive if 
$q \ge t+1$ and $\frac{t-1}{t} \Big(1 - \frac{1}{q^{t-1}} \Big) + \omega < \frac{q-1}{q}$.
This calculation can be further refined because of the additional
constraints on the parameters $\alpha,\beta,\gamma$ mentioned above.

\remove{
\section{Concluding Remarks}

We introduced a new class of problems involving two-level codes,
where the licensed users are 
organized in several groups. The main advantage of two-level codes
is their ability to partially tolerate coalitions larger than the designed
threshold. For instance, in the case of two-level fingerprinting 
codes, if the coalition size is beyond the designed limit,
then up to a certain larger threshold, the tracing algorithm can 
identify a group containing a pirate. We presented constructions
of codes with the two-level fingerprinting property. Our main focus was 
on the narrow-sense rule (\ref{eqn:intro-wenv}) and minimum distance
based decoding.
The concept of two-level codes raises several new questions. 
A few are identified below.

\begin{openprob}
Construct two-level fingerprinting codes under the wide-sense rule
(\ref{eqn:intro-wenv}) for generating forgeries. In particular, 
constructions for binary codes and other tracing algorithms
besides MD decoding are especially of interest.
\end{openprob}

\begin{openprob}
Find upper bounds on the rates of two-level fingerprinting
(also TA, frameproof) codes.
\end{openprob}
}

\emph{Acknowledgment:}
This work was partially supported by NSF through grants CCF0635271, CCF0830699,
and DMS0807411.

\end{document}